\documentclass[12pt]{article}

\usepackage{hyperref}

\usepackage{amsmath} 
\usepackage{amssymb}
\usepackage{amsfonts}
\usepackage{amsthm}

\usepackage{footnote}
\usepackage{footmisc}

\newcommand{\astfootnote}[1]{%
	\let\oldthefootnote=\thefootnote%
	\setcounter{footnote}{0}%
	\renewcommand{\thefootnote}{\fnsymbol{footnote}}%
	\footnote{#1}%
	\let\thefootnote=\oldthefootnote%
}

\newtheorem{theorem}{Theorem}
\newtheorem{lemma}{Lemma}
\newtheorem{corollary}{Corollary}

\theoremstyle{definition}

\begin{document}
	\begin{center}
		\Large
	\textbf{Higher order moments dynamics for some multimode  quantum master equations}\astfootnote{This work is supported by the Russian Science Foundation under grant 17-71-20154.}	
	
		\large 
		\textbf{Iu.A. Nosal}\footnote{Department of Mathematical Physics, Steklov Mathematical Institute of Russian Academy of Sciences	ul. Gubkina 8, Moscow, 119991 Russia\\ E-mail: \href{mailto:nosal.ia16@physics.msu.ru}{nosal.ia16@physics.msu.ru}}\textbf{,}
		\textbf{A.E. Teretenkov}\footnote{Department of Mathematical Methods for Quantum Technologies, Steklov Mathematical Institute of Russian Academy of Sciences,
			ul. Gubkina 8, Moscow 119991, Russia\\ E-mail:\href{mailto:taemsu@mail.ru}{taemsu@mail.ru}}
		\end{center}
		
			\footnotesize
			We derive Heisenberg equations  for arbitrary high order moments of creation and annihilation operators in the case of the quantum master equation with a multimode generator which is quadratic in creation and annihilation operators and obtain their solutions. Based on them we also derive similar equations for the case of the quantum master equation, which occur after averaging the dynamics  with a quadratic generator with respect to the classical Poisson process. This allows us to show that dynamics of  arbitrary finite-order  moments of creation and annihilation operators is fully defined by finite number of linear differential equations in this case.
			\normalsize

	\section{Introduction}
	
	This paper is further development of the series of works \cite{Ter19, Teretenkov20, NosTer20, Linowski2021, Ivanov2022} devoted to quantum master equations for which dynamics of moments can be obtained explicitly. Namely, in \cite{Teretenkov20, NosTer20} we have shown that the master equations, which arise as averaging of unitary evolution with a quadratic (in creation and annihilation operators) Hamiltonian with respect to classical Poisson process, lead to closed linear ordinary differential equations for moments of creation and annihilation operators of fixed order. In \cite{Ivanov2022} something similar was done for the Wiener stochastic processes. The main of this article is to generalize the results of \cite{Teretenkov20} and consider the Gorini-Kossakowski-Sudarshan-Lindblad (GKSL) generators which arise as averaging of initially non-unitary dynamics by the classical Poisson process.   This initial non-unitary dynamics is also assumed to have a quadratic generator in creation and annihilation operators. In contrast to \cite{Teretenkov20, NosTer20, Ivanov2022} we obtain the closed linear ordinary differential equations for the moments of creation and annihilation operators up to fixed order, not for fixed order moments only. Thus, in this case the equations for the moments of creation and annihilation operators of given order depend on solutions of equations for lower-order moments. So they can be solved iteratively. To derive explicit form of such equations we consider several auxiliary problems, which we think can be interesting by themselves.
	
	In section \ref{sec:Leibniz} we derive a formula which allows one to calculate a GKSL generator of a product of operators. It can be viewed as a  generalization of the Leibniz formula for a commutator, so we call it the Leibniz type formula. 
	
	In section \ref{sec:Heisenberg} we use this formula to derive the Heisenberg equations in the case of GKSL generators, which are quadratic in bosonic creation and annihilation operators. Despite the fact that these generators are well studied \cite{bausch1967description, Vanheuverzwijn1978, dodonov1985quantum, DodMan86, Prosen10, Heinosaari2010}, we have not found the explicit form of the Heisenberg equations for higher-order (third and higher) moments in literature in this case. On the one hand, there are some reasons for it. Namely, for Gaussian states the higher moments could be defined by the Isserlis --- Wick theorem (as we discuss in section \ref{sec:Isserlis}). Moreover, even for non-Gaussian states the dynamics is fully defined by dynamics of first and second moments as a special case of quantum Gaussian channels \cite[Sec.~12.4]{Holevo2012}. On the other hand, for our purposes we need these general formulae as an auxiliary result  (in particular, we  use them in section \ref{sec:Poisson} in the case, when the Isserlis --- Wick theorem is not applicable) and the higher moments are important for calculation of covariances of such widespread properties as intensities of electromagnetic field \cite[Subsec.~12.12.2]{Mandel1995}, so its explicit description could be useful even for quadratic generators. 
	
	In section \ref{sec:HeisSol} we derive the solutions of such equations. As the form of the solution is not simplified much in the case of time-independent coefficients we assume a general time-dependent case in this section despite the fact that it is not necessary for the main aim of our article. Moreover, this general time-independent situation could be interesting for the purposes of incoherent  control \cite{Morzhin2019, Morzhin2021} and tomography in the case of time-dependent generators \cite{Manko2022}. Let us remark that although the GKSL equations with a quadratic generator are well-studied, there is still keen current interest in further reasearch into their properties and applications   \cite{Agredo2021a, Agredo2021b, Barthel2021, Gaidash21, Medvedeva21, Bolamos2021, Flynn2021}.
	
	In section \ref{sec:Isserlis} we prove the Isserlis-Wick theorem in our notation. Despite the fact that this section mostly just checks our results form the previous sections, namely, we testify their consistency with this theorem, there is still interest in the Wick theorem in literature \cite{Ferialdi2021} as well as in related combinatorial aspects \cite{Schork2021}. So possibly our results can be useful for such a discussion.  There are some generalizations of the Isserlis theorem in classical probability \cite{Michalowicz2011, Vignat2012} and our results could be relevant for their quantum generalizations.
	
	In section \ref{sec:Poisson} we derive the Heisenberg equations for products of creation and annihilation operators in case of generators which occur after averaging the dynamics  with a quadratic generator with respect to the classical Poisson process. For this case we show that dynamics of  arbitrary finite-order  moments of creation and annihilation operators is fully defined by a finite number of linear differential equations.
	
	In Conclusions we discuss possible generalizations and directions for further development.
	
	\section{Leibniz type formula for GKSL generator}
	\label{sec:Leibniz}
	
	We consider the  GKSL equation for the density matrix of the form
	\begin{equation*}
		\frac{d}{dt} \rho(t) = \mathcal{L}(\rho(t)) ,
	\end{equation*}
	where
	\begin{equation*}
		\mathcal{L}(\rho) =  -i [\hat{H}, \rho] + \sum_j \left(\hat{C}_j \rho \hat{C}_j^{\dagger} - \frac12 \hat{C}_j^{\dagger}\hat{C}_j \rho - \frac12 \rho \hat{C}_j^{\dagger}\hat{C}_j  \right).
	\end{equation*}
	Let us remark that following the tradition (see e.g. \cite{Agredo2021b}) we call it a GKSL equation for unbounded generators too despite the fact that theorems from \cite{Lindblad76, Gorini76} are not applicable in this case and moreover other generators could occur \cite{Holevo2018} in the unbounded case. Let also us remark that we consider the unbounded operators only formally.
	
	In the Heisenberg representation
	\begin{equation*}
		\frac{d}{dt} X(t) = \mathcal{L}^*(X(t)) ,
	\end{equation*}
	where
	\begin{equation}\label{eq:GKSLHeisenberg}
		\mathcal{L}^*(X) = i [\hat{H}, X] + \frac12\sum_j ([\hat{C}_j^{\dagger}, X] \hat{C}_j + \hat{C}_j^{\dagger} [X, \hat{C}_j]).
	\end{equation}
	
	If $ \hat{C}_j = 0 $, then one has a Leibniz rule for a commutator \cite[Sec.~40]{Prasolov1994}
	\begin{equation*}
		\mathcal{L}^*(XY) =X\mathcal{L}^*(Y) + \mathcal{L}^*(X) Y
	\end{equation*}
	and it can be generalized to
	\begin{equation}\label{eq:LeibnizComm}
		\mathcal{L}^*(X_1 \ldots X_m) =X\mathcal{L}^*(Y) + \mathcal{L}^*(X) Y =  \sum_{k=1}^m X_1 \ldots X_{k-1}\mathcal{L}^*(X_{k}) X_{k+1} \ldots X_{m}.
	\end{equation}
	
	For $ \hat{C}_j  \neq 0$ corrections to  Leibniz rule occur \cite[p. 85]{Holevo2002}
	\begin{equation}\label{lem:LeibnizTypeTwoFold}
		\mathcal{L}^*(XY) =X\mathcal{L}^*(Y) + \mathcal{L}^*(X) Y  +  \sum_j [\hat{C}_j^{\dagger}, X][Y,\hat{C}_j ].
	\end{equation}
	
	Let us generalize formula \eqref{eq:LeibnizComm} for this case.
	\begin{lemma}\label{lem:LeibnizType}
		Let $ \mathcal{L}^* $ be defined by \eqref{eq:GKSLHeisenberg}, then
		\begin{align}
			\mathcal{L}^*(X_1 \ldots X_m) =& \sum_{k=1}^m X_1 \ldots X_{k-1}\mathcal{L}^*(X_{k}) X_{k+1} \ldots X_{m} \nonumber \\
			&+  \sum_{1 \leqslant k <l \leqslant m}\sum_j X_1 \ldots X_{k-1} [\hat{C}_j^{\dagger}, X_k ]  X_{k+1} \ldots X_{l-1} [X_{l}, \hat{C}_j ] X_{l+1} \ldots X_m. \label{eq:LeibnizType}
		\end{align}
	\end{lemma}

	\begin{proof}
		Let us prove formula \eqref{eq:LeibnizType} by induction. Eq.~\eqref{lem:LeibnizTypeTwoFold} is the base of induction. Let us assume that \eqref{eq:LeibnizType} is proved for $ m $ and let us prove it for $ m + 1 $. First of all let us apply \eqref{lem:LeibnizTypeTwoFold} assuming $ X = X_1 \ldots X_{m} $ and $ Y = X_{m+1} $, then we have
		\begin{equation*}
			\mathcal{L}^*(X_1 \ldots X_{m+1}) = X_1 \ldots X_{m} \mathcal{L}^*(X_{m+1}) + \mathcal{L}^*(X_1 \ldots X_{m} ) X_{m+1}  +  \sum_j [\hat{C}_j^{\dagger}, X_1 \ldots X_{m}][X_{m+1}, \hat{C}_j ].
		\end{equation*}
		Now let us apply Eq.~\eqref{eq:LeibnizType} for the $ m $-fold case and the Leibniz rule for a commutator, then we have
		\begin{align*}
			\mathcal{L}^*(X_1 \ldots X_{m+1}) =  X_1 \ldots X_{m} \mathcal{L}^*(X_{m+1}) + \sum_{k=1}^m X_1 \ldots \mathcal{L}^*(X_{k}) \ldots X_{m} X_{m+1}\\
			+ \sum_{1 \leqslant k <l \leqslant m}\sum_j X_1 \ldots X_{k-1} [\hat{C}_j^{\dagger}, X_k ]  X_{k+1} \ldots X_{l-1} [X_{l}, \hat{C}_j ] X_{l+1} \ldots X_m X_{m+1}\\
			+ \sum_{k=1}^m \sum_j X_{1} \ldots X_{k-1} [\hat{C}_j^{\dagger}, X_k ] X_{k+1} \ldots X_{m}[X_{m+1}, \hat{C}_j ]\\
			= \sum_{k=1}^{m+1} X_1 \ldots X_{k-1}\mathcal{L}^*(X_{k}) X_{k+1} \ldots X_{m+1}  \\
			+  \sum_{1 \leqslant k <l \leqslant m+1}\sum_j X_1 \ldots X_{k-1} [\hat{C}_j^{\dagger}, X_k ]  X_{k+1} \ldots X_{l-1} [X_{l}, \hat{C}_j ] X_{l+1} \ldots X_{m + 1}. 
		\end{align*}
		Thus, we obtain Eq.~\eqref{eq:LeibnizType} for the $ (m+1) $-fold case.
	\end{proof}
	
	\section{Heisenberg equations}
	\label{sec:Heisenberg}
	
	We need to briefly present some notation from \cite{Ter19, Ter16, Ter17a, Ter19R} to formulate our result. In this section we consider Hilbert space $\otimes_{j=1}^n\ell_2$.  Let us define the $2n$-dimensional vector of annihilation and creation operators $\mathfrak{a} = (\hat{a}_1, \ldots, \hat{a}_n, \hat{a}_1^{\dagger}, \ldots, \hat{a}_n^{\dagger} )^T$ satisfying canonical commutation relations \cite[Sec. 1.1.2]{scalli2003} $ [\hat{a}_i, \hat{a}_j^{\dagger}] = \delta_{ij}$, $ [\hat{a}_i, \hat{a}_j] = [\hat{a}_i^{\dagger}, \hat{a}_j^{\dagger}]= 0 $. One could denote linear and quadratic forms in such operators by $ f^T \mathfrak{a}  $ and $ \mathfrak{a}^T K \mathfrak{a} $, respectively. Here $  f \in \mathbb{C}^{2n} $ and $ K \in \mathbb{C}^{2n \times 2n} $. Let us define $2n \times 2n$-dimensional matrices as
	\begin{equation*}
		J = \biggl(
		\begin{array}{cc}
			0 & -I_n \\ 
			I_n & 0
		\end{array} 
		\biggr), \qquad
		E = \biggl(
		\begin{array}{cc}
			0 & I_n \\ 
			I_n & 0
		\end{array} 
		\biggr),
	\end{equation*}
	where $ I_n $ --- identity matrix from $ \mathbb{C}^{n \times n} $. Then the canonical commutation relations take the form
	\begin{equation}\label{eq:commRel}
		[f^T \mathfrak{a},\mathfrak{a}^T g]  = - f^T J g, \qquad \forall g, f \in \mathbb{C}^{2n}.
	\end{equation}
	
	We also define the $\sim$-conjugation of vectors and matrices by the formulae
	\begin{equation*}
		\tilde{g} = E\overline{g}, \; g \in \mathbb{C}^{2n}, \qquad \tilde{K} = E \overline{K} E, \; K \in \mathbb{C}^{2n \times 2n},
	\end{equation*}
	where the overline is an (elementwise) complex conjugation.
	
	\begin{lemma}
		For $ \hat{H} = \frac{1}{2} \mathfrak{a}^T H \mathfrak{a} + f^T \mathfrak{a}  $, where $ H = H^T = \tilde{H} $ ,  $ f = \tilde{f} $,  and $ \hat{C}_j = \gamma_j^T \mathfrak{a}  $  Eq.~\eqref{eq:GKSLHeisenberg} takes the form
		\begin{equation}\label{eq:quadGenHeis}
			\mathcal{L}_{H, \Gamma, f}^*(X) \equiv i\left[\frac{1}{2} \mathfrak{a}^T H \mathfrak{a} + f^T \mathfrak{a} , X \right] +  \mathfrak{a}^T X \Gamma^T \mathfrak{a}  - \frac{1}{2} \mathfrak{a}^T \Gamma^T \mathfrak{a} X  - \frac{1}{2} X \; \mathfrak{a}^T \Gamma^T \mathfrak{a} ,
		\end{equation}
		where $ \Gamma \equiv \sum\limits_j \gamma_j \tilde{\gamma}_j^T $.
	\end{lemma}
	
	Let us formulate a part of lemma 10 from \cite{Ter19R} which is needed for purposes of this paper.
	\begin{lemma}\label{lem:applGenToFistOrder}
		Let $ \mathcal{L}_{H, \Gamma, f}^* $ be defined by \eqref{eq:quadGenHeis}, then 
		\begin{equation}\label{eq:linearChange}
			\mathcal{L}_{H, \Gamma, f}^* ( \mathfrak{a} ) = B  \mathfrak{a} + \varphi,
		\end{equation}
		where
		\begin{equation}\label{eq:BPhiDef}
			B \equiv J \left(i H + \frac{\Gamma^T - \Gamma}{2}\right), \qquad \varphi \equiv i J f
		\end{equation}
		and $ \mathcal{L}_{H, \Gamma, f}^* ( \mathfrak{a} ) $ is understood as application of $ \mathcal{L}_{H, \Gamma, f}^* $ to each element of $ \mathfrak{a} $.
	\end{lemma}
	
	Now let us generalize it to the $ m $-fold case.
	
	\begin{lemma}
		Let $ g_1, \ldots, g_m \in \mathbb{C}^{2n} $, then
		\begin{align}
			\mathcal{L}_{H, \Gamma, f}^*(g_1^T \mathfrak{a} \ldots g_m^T \mathfrak{a}) &=\sum_{k=1}^m g_1^T \mathfrak{a}\ldots g_{k-1}^T \mathfrak{a} g_k^T (B\mathfrak{a} + \varphi) g_{k+1}^T \mathfrak{a} \ldots g_m^T \mathfrak{a} \nonumber\\
			& +  \sum_{1 \leqslant k <l \leqslant m}  (g_{k}^T \Xi g_{l}) g_1^T \mathfrak{a}  \ldots g_{k-1}^T\mathfrak{a}      g_{k+1}^T\mathfrak{a}\ldots  g_{l-1}^T\mathfrak{a}  g_{l+1}^T\mathfrak{a} \ldots  g_{m}^T\mathfrak{a}, \label{eq:applQuadGenToProduct}
		\end{align}
		where
		\begin{equation}\label{eq:XiDef}
			\Xi \equiv J \Gamma^T J.
		\end{equation}
	\end{lemma}
	
	\begin{proof}
		By lemmas \ref{lem:LeibnizType} and \ref{lem:applGenToFistOrder} and taking into account \eqref{eq:commRel} we have
		\begin{align*}
			\mathcal{L}^*(g_1^T \mathfrak{a} \ldots g_m^T \mathfrak{a}) =& \sum_{k=1}^m g_1^T \mathfrak{a}\ldots g_{k-1}^T \mathfrak{a} \mathcal{L}^*(g_k^T \mathfrak{a}) g_{k+1}^T \mathfrak{a} \ldots g_m^T \mathfrak{a} \\
			&+  \sum_{1 \leqslant k <l \leqslant m}\sum_i g_1^T \mathfrak{a}  \ldots g_{k-1}^T\mathfrak{a}  [\hat{C}_i^{\dagger},  g_{k}^T\mathfrak{a}]   g_{k+1}^T\mathfrak{a}\ldots  g_{l-1}^T\mathfrak{a} [ g_{l}^T\mathfrak{a}, \hat{C}_i ]  g_{l+1}^T\mathfrak{a} \ldots  g_{m}^T\mathfrak{a} \\
			=& \sum_{k=1}^m g_1^T \mathfrak{a}\ldots g_{k-1}^T \mathfrak{a} g_k^T (B\mathfrak{a} + \varphi) g_{k+1}^T \mathfrak{a} \ldots g_m^T \mathfrak{a}\\
			& +  \sum_{1 \leqslant k <l \leqslant m}\sum_i g_1^T \mathfrak{a}  \ldots g_{k-1}^T\mathfrak{a}   (g_{k}^T J \tilde{\gamma}_i)   g_{k+1}^T\mathfrak{a}\ldots  g_{l-1}^T\mathfrak{a}  (\gamma_i^T J g_{l}) g_{l+1}^T\mathfrak{a} \ldots  g_{m}^T\mathfrak{a} \\
			=& \sum_{k=1}^m g_1^T \mathfrak{a}\ldots g_{k-1}^T \mathfrak{a} g_k^T (B\mathfrak{a} + \varphi) g_{k+1}^T \mathfrak{a} \ldots g_m^T \mathfrak{a}\\
			& +  \sum_{1 \leqslant k <l \leqslant m}  (g_{k}^T \Xi g_{l}) g_1^T \mathfrak{a}  \ldots g_{k-1}^T\mathfrak{a}      g_{k+1}^T\mathfrak{a}\ldots  g_{l-1}^T\mathfrak{a}  g_{l+1}^T\mathfrak{a} \ldots  g_{m}^T\mathfrak{a}.
		\end{align*}
		Thus, we obtain \eqref{eq:applQuadGenToProduct}.
	\end{proof}
	
	Let us write in more compact notation. One can think about $ g_1^T \mathfrak{a} \ldots g_m^T \mathfrak{a} $ as components of tensor $ \mathfrak{a} \otimes \ldots \otimes \mathfrak{a} $. Similarly to many-body physics set-ups \cite[Subsec. 3.7.2]{Breuer2002} let us use the subscripts for $ \mathfrak{a} $, $ \varphi $, $ \Xi $  to denote the number of the tensor multiplicand to which it corresponds, e.g.  $ \mathfrak{a}_{1} \mathfrak{a}_{2}  \equiv \mathfrak{a} \otimes \mathfrak{a} $, $ \mathfrak{a}_{1} \varphi_{2}  \equiv \mathfrak{a} \otimes \varphi $,  $ \Xi_{12} \mathfrak{a}_3 \equiv \Xi \otimes \mathfrak{a} $ and so on. Similarly, the subscript for $ B $ denotes the number of the tensor multiplicand in which this matrix  acts. Then, Eq. \eqref{eq:applQuadGenToProduct} takes the form
	\begin{equation*}
		\mathcal{L}_{H, \Gamma, f}^* \biggl( \prod_{j=1}^{m} \mathfrak{a}_j \biggr) = \left(\sum_{k=1}^m  B_k\right) \prod_{j=1}^{m} \mathfrak{a}_j  + \sum_{k=1}^m  \varphi_k \prod_{\substack{j=1\\ j \neq k}}^{m} \mathfrak{a}_j  + \sum_{1 \leqslant k <l \leqslant m}  \Xi_{kl}\prod_{\substack{j=1\\ j \neq k, l}}^{m} \mathfrak{a}_j.
	\end{equation*}
	But let us make our notation even more compact by introducing set-valued subscripts
	\begin{equation*}
		\mathfrak{a}_{I} \equiv \prod_{j \in I} \mathfrak{a}_j,
	\end{equation*}
	where $ I $ is some subset of a natural number. Similarly, let us denote by $ p $ a set of two natural numbers (pairs) and let $ P(I) $ be all possible pairs from $ I $. Then  Eq. \eqref{eq:applQuadGenToProduct}  takes the form 
	\begin{equation*}
		\mathcal{L}_{H, \Gamma, f}^* ( \mathfrak{a}_I) = \biggl(\sum_{k \in I}  B_k \biggr) \mathfrak{a}_I + \sum_{k \in I}  \varphi_k \mathfrak{a}_{I\setminus \{k\}} + \sum_{p \in P(I)}  \Xi_{p}  \mathfrak{a}_{I \setminus p}.
	\end{equation*}
	
	Let us consider the simplest cases. Namely, for $ I = \{1\} $ we revisit \eqref{eq:linearChange}. For $ I = \{1,2\} $ we have
	\begin{equation*}
		\mathcal{L}_{H, \Gamma, f}^* ( \mathfrak{a}_{12}) =(B_1 + B_2) \mathfrak{a}_{12} + \varphi_1 \mathfrak{a}_{2} + \varphi_2 \mathfrak{a}_{1}  +   \Xi_{12}  .
	\end{equation*}
	If one arranges the elements of the tensor $ \mathfrak{a}_{12} \equiv \mathfrak{a}_1 \mathfrak{a}_2  \equiv \mathfrak{a} \otimes \mathfrak{a} $ into the matrix $ \mathfrak{a}  \mathfrak{a}^T $, then this equation takes the form
	\begin{equation*}
		\mathcal{L}_{H, \Gamma, f}^* ( \mathfrak{a}  \mathfrak{a}^T) =  B \mathfrak{a}  \mathfrak{a}^T +  \mathfrak{a}  \mathfrak{a}^T B^T + \varphi \mathfrak{a}^T +   \mathfrak{a} \varphi^T  +   \Xi.
	\end{equation*}
	Taking into account \eqref{eq:BPhiDef} and \eqref{eq:XiDef} we have
	\begin{align*}
		\mathcal{L}_{H, \Gamma, f}^* \left( \mathfrak{a} \; \mathfrak{a}^T \right) =  J \left(i H + \frac{\Gamma^T - \Gamma}{2} \right) \mathfrak{a} \; \mathfrak{a}^T + \mathfrak{a} \; \mathfrak{a}^T \left(-i H + \frac{\Gamma^T - \Gamma}{2} \right) J + J \Gamma^T J 
		+ i J f \; \mathfrak{a}^T - i \mathfrak{a} \; f^T J,
	\end{align*}
	which coincides with the other part of lemma 10 from \cite{Ter19R}.
	
	Then the Heisenberg equations for generator \eqref{eq:quadGenHeis} with possibly time-dependent coefficients $ H(t) $, $ f(t) $, $ \Gamma(t) $
	\begin{equation}\label{eq:Heisenberg}
		\frac{d}{dt}\mathfrak{a}_{I} (t)  =\biggl(\sum_{k \in I}  B_k(t) \biggr) \mathfrak{a}_I(t) + \sum_{k \in I}  \varphi_k(t) \mathfrak{a}_{I\setminus \{k\}}(t) + \sum_{p \in P(I)}  \Xi_{p}(t)  \mathfrak{a}_{I \setminus p}(t),
	\end{equation}
	where definitions \eqref{eq:BPhiDef}, \eqref{eq:XiDef} hold the same as for time-independent $ H(t) $, $ f(t) $, $ \Gamma(t) $.
	
	\section{Solution of Heisenberg equations}
	\label{sec:HeisSol}
	
	In this section we solve Eqs.~\eqref{eq:Heisenberg}.
	
	\begin{theorem}\label{th:HeisSol}
		Let $ G(t) $ be a solution of the Cauchy problem
		\begin{equation}\label{eq:propGtEq}
			\frac{d}{dt} G(t) = B(t)  G(t), \qquad G(0) = I_{2n},
		\end{equation}
		and
		\begin{equation}\label{eq:psiGtDef}
			\psi(t)= \int_0^t d \tau \;   (G(\tau) )^{-1} \varphi(\tau), \qquad \beta_{p}(t) = \int_0^t d \tau \;  (G_p(\tau) )^{-1}  \Xi_p(\tau),
		\end{equation}
		where $ p $ is a pair of indices (always in ascending order). Similarly to the previous section a natural subscript $ k  $ for $ \psi_{k}(t) $ and $ G_{k}(t) $ means the number of the tensor multiplicand to which this vector or matrix corresponds. And similarly for  set-valued indices we define
		\begin{equation}\label{eq:psiGtBetaIDef}
			\psi_{I}(t) \equiv \prod_{k \in I} \psi_k(t), \qquad G_{I}(t) \equiv \prod_{k \in I} G_k(t), \qquad \beta_{I} (t) \equiv \sum_{I =p_1 \sqcup \ldots \sqcup p_{|I|/2}} \beta_{p_1}(t) \ldots \beta_{p_{|I|/2}}(t),
		\end{equation}
		where the last sum runs over all pairings of indices for $ I $ if $ |I| $ is even and equals zero for odd $ |I| $.
		
		Then the solution of \eqref{eq:Heisenberg} takes the form
		\begin{equation}\label{eq:HeisenbergSolution}
			\mathfrak{a}_{I} (t) = G_{I}(t) \sum_{I = I_1 \sqcup I_2  \sqcup I_3 } \psi_{I_1} (t) \beta_{I_2} (t) \mathfrak{a}_{I_3} (0) ,
		\end{equation}
		where the sum runs over all possible expansion of $ I $ in a disjoint union of three sets $ I_1, I_2, I_3 $. 
	\end{theorem}
	
	\begin{proof}
		Let us mention that
		\begin{equation*}
			\frac{d}{dt}G_{I}(t) =  \frac{d}{dt}\left( \prod_{l \in I} G_{l}(t)\right) =\sum_{k \in I}  B_{k}(t) \prod_{l \in I} G_{l}(t) = \left(\sum_{k \in I}  B_{k}(t)\right) G_{I}(t) .
		\end{equation*}
		Let us define
		\begin{equation}\label{eq:intPictureTransform}
			\tilde{\mathfrak{a}}_{I} (t) \equiv (G_{I}(t) )^{-1}	\mathfrak{a}_{I} (t), \qquad \tilde{\varphi}_{I} (t) \equiv (G_{I}(t) )^{-1}	\varphi_{I} (t), \qquad  \tilde{\Xi}_{kl}(t)  \equiv(G_{I}(t) )^{-1} \Xi_{kl}(t) .
		\end{equation}
		Then
		\begin{align*}
			\frac{d}{dt} \tilde{\mathfrak{a}}_{I} (t) =& - (G_{I}(t) )^{-1} \left(\frac{d}{dt} G_{I}(t)\right) (G_{I}(t) )^{-1}	\mathfrak{a}_{I} (t) +  (G_{I}(t) )^{-1} \frac{d}{dt} 	\mathfrak{a}_{I} (t)\\
			=& - (G_{I}(t) )^{-1}  \left(\sum_{k \in I}  B_{k}(t)\right) 	\mathfrak{a}_{I} (t) + (G_{I}(t) )^{-1}\biggl(\sum_{k \in I}  B_k(t) \biggr) \mathfrak{a}_I(t) \\
			&+ (G_{I}(t) )^{-1}\sum_{k \in I}  \varphi_k(t) \mathfrak{a}_{I\setminus \{k\}}(t) + (G_{I}(t) )^{-1} \sum_{p \in P(I)}  \Xi_{p}(t)  \mathfrak{a}_{I \setminus p}(t) \\
			=& \sum_{k \in I} (G_{k}(t) )^{-1} \varphi_k(t) (G_{I\setminus \{k\}}(t))^{-1} \mathfrak{a}_{I\setminus \{k\}}(t) +\sum_{p \in P(I)}   (G_{p}(t) )^{-1} \Xi_{p}(t)   (G_{I \setminus p}(t) )^{-1}  \mathfrak{a}_{I \setminus p}(t)\\
			=& \sum_{k \in I} \tilde{\varphi}_k(t) \tilde{\mathfrak{a}}_{I\setminus \{k\}}(t) +\sum_{p \in P(I)} \tilde{\Xi}_{p}(t)  \tilde{\mathfrak{a}}_{I \setminus p}(t)
		\end{align*}
		By integrating with respect to $ t $ we obtain a recurrent equation for $ \tilde{\mathfrak{a}}_{I} (t) $ (in terms of $ \tilde{\mathfrak{a}}_{I'} (t) $ with $ I' $ of lower cardinality than $ I $)
		\begin{equation}\label{eq:reccurentRel}
			\tilde{\mathfrak{a}}_{I} (t)  = \tilde{\mathfrak{a}}_{I} (0) + \sum_{k \in I} \int_0^t d \tau \;   \tilde{\varphi}_k(\tau) \tilde{\mathfrak{a}}_{I \setminus \{k\}} (\tau) + \sum_{p \in P(I)}  \int_0^t d \tau \;  \tilde{\Xi}_{p}(\tau) \tilde{\mathfrak{a}}_{I \setminus p} (\tau) .
		\end{equation}
		Let us prove that its solution has the form
		\begin{equation}\label{eq:tildeSolutions}
			\tilde{\mathfrak{a}}_{I} (t)  = \sum_{I = I_1 \sqcup I_2  \sqcup I_3 } \psi_{I_1} (t) \beta_{I_2} (t) \tilde{\mathfrak{a}}_{I_3} (0) 
		\end{equation}
		by induction. The base of induction
		\begin{equation*}
			\tilde{\mathfrak{a}}_{l} (t)  = \tilde{\mathfrak{a}}_{l} (0) + \int_0^t d \tau \;   \tilde{\varphi}_l(\tau) = \tilde{\mathfrak{a}}_{l} (0)  + \psi_{l} (t) = \sum_{\{l\} = I_1 \sqcup I_3 }  \psi_{I_1} (t) \tilde{\mathfrak{a}}_{I_3} (0).
		\end{equation*}
		So let \eqref{eq:tildeSolutions} be proved  for $ \tilde{\mathfrak{a}}_{I \setminus \{k\}} (t)  $ and $ \tilde{\mathfrak{a}}_{I \setminus p} (t) $ and let us prove it for $ \tilde{\mathfrak{a}}_{I} (t)  $. Namely, by Eq.~\eqref{eq:reccurentRel} we have
		\begin{align*}
			\tilde{\mathfrak{a}}_{I} (t)  = \tilde{\mathfrak{a}}_{I} (0) &+ \sum_{k \in I} \sum_{I\setminus \{k\} = I_1 \sqcup I_2  \sqcup I_3 } \int_0^t d \tau \;   \tilde{\varphi}_k(\tau)  \psi_{I_1 }(\tau) \beta_{I_2}(\tau) \tilde{\mathfrak{a}}_{I_3} (0) \\
			&+ \sum_{p \in P(I)}  \int_0^t d \tau \;  \tilde{\Xi}_{p}(\tau) \sum_{I\setminus p = I_1 \sqcup I_2  \sqcup I_3} \psi_{I_1} (\tau) \beta_{I_2} (\tau) \tilde{\mathfrak{a}}_{I_3} (0) .
		\end{align*}
		Let us remark that
		\begin{equation*}
			\frac{d}{dt} \psi_{I}(t) = \sum_{k \in I}  \tilde{\varphi}_k(t) \psi_{I \setminus \{k\}}(t) 
		\end{equation*}
		and
		\begin{align*}
			\frac{d}{dt} \beta_{I}(t) &=  \sum_{I =p_1 \sqcup \ldots \sqcup p_{|I|/2}} \frac{d}{dt}(\beta_{p_1}(t) \ldots \beta_{p_{|I|/2}}(t)) \\
			&= \sum_p \left(\frac{d}{dt} \beta_{p}(t)\right) \sum_{I \setminus p =p_1 \sqcup \ldots \sqcup p_{|I|/2-1}} \beta_{p_1}(t) \ldots \beta_{p_{|I|/2-1}}(t) = \sum_{p } \tilde{\Xi}_p(t) \beta_{I \setminus p}(t)
		\end{align*}
		for even $ |I| $ and the sum is taken over all possible pairings (matchings). Thus, we have
		\begin{align*}
			\tilde{\mathfrak{a}}_{I} (t)  =& \tilde{\mathfrak{a}}_{I} (0) + \sum_{I = I_1 \sqcup I_2  \sqcup I_3 } \int_0^t d \tau \;  \left(\frac{d}{d \tau}\psi_{I_1}(\tau)\right) \beta_{I_2}(\tau) \tilde{\mathfrak{a}}_{I_3} (0) \\
			&+  \int_0^t d \tau \;  \sum_{I = I_1 \sqcup I_2  \sqcup I_3} \psi_{I_1} (\tau) \frac{d}{d \tau} \beta_{I_2} (\tau) \tilde{\mathfrak{a}}_{I_3} (0) = \sum_{I = I_1 \sqcup I_2  \sqcup I_3 } \psi_{I_1} (t) \beta_{I_2} (t) \tilde{\mathfrak{a}}_{I_3} (0) .
		\end{align*}
		Thus, we have proved \eqref{eq:tildeSolutions}. Taking into account \eqref{eq:intPictureTransform} we have \eqref{eq:HeisenbergSolution}.
	\end{proof}
	
	Let us consider several special cases to illustrate formula \eqref{eq:HeisenbergSolution}. For $ I = \{1\} $ we have
	\begin{equation*}
		\mathfrak{a}_{1} (t) = G_{1}(t) ( \mathfrak{a}_{1} (0) + \psi_1(t))
	\end{equation*}
	If one defines the vector of first moments $ \mu(t) \equiv \langle \mathfrak{a}_{1} (t) \rangle $, where $ \langle \; \cdot \; \rangle \equiv \operatorname{tr}(\; \cdot \; \rho_0) $ is the average taken with respect to the initial density matrix, then one has
	\begin{equation}\label{eq:meanEvol}
		\mu_1(t) = G_{1}(t) ( \mu_1(0) + \psi_1(t)).
	\end{equation}
	For $ I = \{1,2\} $ we have
	\begin{equation*}
		\mathfrak{a}_{12} (t) =  G_{12}(t) (\mathfrak{a}_{12} (0) + \psi_1(t) \mathfrak{a}_{2} (0) + \psi_2(t)  \mathfrak{a}_{1} (0) + \psi_1(t) \psi_2(t)   +  \beta_{12}(t)).
	\end{equation*}
	After averaging  with respect to the initial density matrix we have
	\begin{equation}\label{eq:secondMomEvol}
		\langle 	\mathfrak{a}_{12} (t) \rangle =  G_{12}(t) (\langle \mathfrak{a}_{12} (0) \rangle + \psi_1(t) \langle \mathfrak{a}_{2} (0 \rangle) + \psi_2(t)  \langle \mathfrak{a}_{1} (0) \rangle + \psi_1(t) \psi_2(t)   +  \beta_{12}(t)).
	\end{equation}
	
	If one defines the matrix of second central moments similarly to \cite[Def. 5]{Ter19R}  as
	\begin{equation*}
		D_{12}(t) \equiv  \langle 	\mathfrak{a}_{12} (t) \rangle - \langle \mathfrak{a}_{1} (t) \rangle \langle \mathfrak{a}_{2} (t) \rangle,
	\end{equation*}
	then we have
	\begin{equation}\label{eq:centMomEvol}
		D_{12}(t) =  G_{12}(t) ( D_{12}(0) +  \beta_{12}(t)).
	\end{equation}
	Eqs. \eqref{eq:meanEvol} and \eqref{eq:centMomEvol} coincide with \cite[Prop. 7]{Ter19R}. 
	
	For simplicity let us now assume $ \varphi = 0 $, then we have
	\begin{align*}
		\mathfrak{a}_{1234} (t)  = G_{1234}(t)( \mathfrak{a}_{1234} (0) &+  \mathfrak{a}_{12} (0) \beta_{34}(t)  + \mathfrak{a}_{13} (0) \beta_{24}(t) + \mathfrak{a}_{14} (0) \beta_{23}(t)\\
		&+\beta_{12}(t) \beta_{34}(t) +   \beta_{13} (0) \beta_{24}(t)  + \beta_{13} (0) \beta_{24}(t))
	\end{align*}
	Averaging  with respect to the initial density matrix and taking into account Eq.~\eqref{eq:secondMomEvol} we have
	\begin{align*}
		\langle \mathfrak{a}_{1234} (t) \rangle  - \langle \mathfrak{a}_{12} (t) \rangle   \langle \mathfrak{a}_{34} (t) \rangle = G_{1234}(t) (	\langle \mathfrak{a}_{1234} (0) \rangle  - \langle \mathfrak{a}_{12} (0) \rangle   \langle \mathfrak{a}_{34} (0) \rangle  \\
		+ \langle \mathfrak{a}_{13} (0) \rangle  \beta_{24}(t) + \langle \mathfrak{a}_{14}  (0)\rangle \beta_{23}(t) +   \beta_{13} (0) \beta_{24}(t)  + \beta_{13} (0) \beta_{24}(t) ).
	\end{align*}
	In particular, such a tensor contains the terms of the form $ \langle \hat{a}_{i}^{\dagger} \hat{a}_{i} \hat{a}_{j}^{\dagger} \hat{a}_{j}  \rangle -  \langle \hat{a}_{i}^{\dagger} \hat{a}_{i}  \rangle \langle\hat{a}_{j}^{\dagger} \hat{a}_{j} \rangle $, which describe the correlations of intensities of electromagnetic field \cite[Subsec.~12.12.2]{Mandel1995}.
	
	Let us consider the case of constant coefficients. Then the solution of Eq.~\eqref{eq:propGtEq} has the form 
	\begin{equation}\label{eq:GtConst}
		G(t) = e^{B t}.
	\end{equation}
	Then Eqs.~\eqref{eq:psiGtDef} take the form
	\begin{align}
		\psi(t) &= \int_0^t d \tau \;    e^{ -B \tau}  \varphi = \frac{1 - e^{-B t}}{B} \varphi, \label{eq:psitConst}  \\
		\beta_{12}(t) &= \int_0^t d \tau \;   e^{ -(B_1 + B_2) \tau}   \Xi_{12} = \frac{1 - e^{-(B_1 + B_2)  t}}{B_1 + B_2}  \Xi_{12}. \label{eq:XitConst}
	\end{align}
	Here and below we understand function of $ B $ as a Taylor series
	\begin{equation*}
		\frac{1 - e^{-B t}}{B} = \sum_{j=1}^{\infty} (-1)^{j+1} \frac{1}{j!} B^{j-1} t^j.
	\end{equation*}
	In particular, it is well-defined even if $ B $ is degenerate. Then Eqs. \eqref{eq:meanEvol} take the form
	\begin{equation*}
		\mu_{1} (t) = e^{B_1 t} \mu_{1} (0) + \frac{e^{B_1 t} - 1}{B_1} \varphi, \qquad D_{12} (t) = e^{(B_1 + B_2) t} D_{12} (0) + \frac{e^{(B_1 + B_2) t} - 1}{B_1 + B_2} \beta_{12},
	\end{equation*}
	which coincides with \cite[Prop. 9]{Ter19R}.

	\section{Consistency with Isserlis - Wick theorem}
	\label{sec:Isserlis}
	
	Several different (but deeply related) statements are usually called the Wick theorem \cite[Parargraph 17]{Bogolubov2005}. Namely, some of them are about normal order of operators and valid independently of state. The other ones are about higher moments of creation and annihilation operators for Gaussian states (often about only some special cases of them like vaccuum or thermal states for quadratic Hamiltonians). So the latter ones are quantum versions of the Isserlis theorem for higher moments of classical Gaussian states. To highlight that we speak about the statement of the second kind we refer to it as the Isserlis - Wick theorem.
	
	\begin{lemma}
		Let $ g \in \mathbb{C}^{2n} $ and $ \mathbf{z}  \in \mathbb{C}^{2n}  $ have the form $ \mathbf{z} = (z_1, \ldots, z_n, \overline{z}_1, \ldots, \overline{z}_n)^T $, then
		\begin{equation}\label{eq:productOfLienarAndExponential}
			g^T\mathfrak{a} e^{i \mathbf{z}^T \mathfrak{a}}  = g^T\left(	\frac{\partial}{\partial(i \mathbf{z})}  -  \frac12 J (i \mathbf{z})  \right)  e^{i \mathbf{z}^T \mathfrak{a}}.
		\end{equation}
	\end{lemma}
	
	\begin{proof}
		By Feynman-Wilcox formula \cite{Chebotarev12}
		\begin{align*}
			\frac{\partial}{\partial(i \mathbf{z})} e^{i \mathbf{z}^T \mathfrak{a}} = \int_0^1 ds e^{i \mathbf{z}^T \mathfrak{a} s}	\frac{\partial}{\partial(i \mathbf{z})} (i \mathbf{z}^T \mathfrak{a}) e^{-i \mathbf{z}^T \mathfrak{a} s} e^{i \mathbf{z}^T \mathfrak{a}} = \int_0^1 ds e^{i \mathbf{z}^T \mathfrak{a} s}	 \mathfrak{a} e^{-i \mathbf{z}^T \mathfrak{a} s} e^{i \mathbf{z}^T \mathfrak{a}} \\
			=\int_0^1 ds \left(\mathfrak{a} + s J (i \mathbf{z}) \right)  e^{i \mathbf{z}^T \mathfrak{a}} = \left(\mathfrak{a} + \frac12 J (i \mathbf{z}) \right)  e^{i \mathbf{z}^T \mathfrak{a}},
		\end{align*}
		where $ e^{i \mathbf{z}^T \mathfrak{a} s}	 \mathfrak{a} e^{-i \mathbf{z}^T \mathfrak{a} s} = \mathfrak{a} + s J (i \mathbf{z})  $ due to \cite[Lemma 4]{Ter19R}. Multiplying by $ g $ we obtain \eqref{eq:productOfLienarAndExponential}.
	\end{proof}
	
	Applying Eq. \eqref{eq:productOfLienarAndExponential} iteratively we obtain the following lemma.
	\begin{lemma}
		\label{lem:prodMultExp}
		Let $ g_1, \ldots, g_m \in \mathbb{C}^{2n} $ and $ \mathbf{z}  \in \mathbb{C}^{2n}  $ have the form $ \mathbf{z} = (z_1, \ldots, z_n, \overline{z}_1, \ldots, \overline{z}_n)^T $, then
		\begin{equation}
			\label{eq:prodMultExp}
			g_1^T\mathfrak{a} \ldots g_m^T\mathfrak{a} e^{i \mathbf{z}^T \mathfrak{a}}  = g_m^T\left(	\frac{\partial}{\partial(i \mathbf{z})}  -  \frac12 J (i \mathbf{z})  \right) \ldots  g_1^T\left(	\frac{\partial}{\partial(i \mathbf{z})}  -  \frac12 J (i \mathbf{z})  \right)  e^{i \mathbf{z}^T \mathfrak{a}}.
		\end{equation}
	\end{lemma}
	
	The definition \cite[Subsec.~12.3.1]{Holevo2012} of the characteristic function of the state $ \rho $ in our notation takes the form 
	\begin{equation}\label{eq:charFun}
		h(\mathbf{z}) \equiv \operatorname{tr} (e^{i \mathbf{z}^T \mathfrak{a}} \rho), \qquad \mathbf{z} = (z_1, \ldots, z_n, \overline{z}_1, \ldots, \overline{z}_n)^T.
	\end{equation}
	
	Lemma \ref{lem:prodMultExp} allows one to calculate moments of creation and annihilation operators in terms of $ h(\mathbf{z}) $, namely,  averaging both sides of Eq.~\eqref{eq:prodMultExp} we obtain the following corollary.
	\begin{corollary}
		\label{cor:meanMomets}
		Let $ g_1, \ldots, g_m \in \mathbb{C}^{2n} $ and $ h(\mathbf{z}) $ be defined by Eq.~\eqref{eq:charFun}
		\begin{equation*}
			\langle g_1^T\mathfrak{a} \ldots g_m^T\mathfrak{a}  \rangle  = \left. g_m^T\left(	\frac{\partial}{\partial(i \mathbf{z})}  -  \frac12 J (i \mathbf{z})  \right) \ldots  g_1^T\left(	\frac{\partial}{\partial(i \mathbf{z})}  -  \frac12 J (i \mathbf{z})  \right)  h(\mathbf{z}) \right|_{\mathbf{z} = 0},
		\end{equation*}
		where $ \langle \; \cdot \; \rangle \equiv \operatorname{tr}(\; \cdot \; \rho_0) $.
	\end{corollary}
	
	Now let us consider a Gaussian state with zero mean and the covariance matrix $ C $, i.e. let us assume \cite[Subsec.~12.3.2]{Holevo2012} $ h(\mathbf{z}) = e^{\frac12(i \mathbf{z})^T C (i \mathbf{z})}  $. The covariance matrix $ C $ is a symmetric part of the matrix of second central moments \cite[Sec.~3.3]{Ter19R}
	\begin{equation*}
		C = \frac{1}{2}(D + D^T),
	\end{equation*}
	Due to the fact that the skew-symmetric part of the matrix $ D $ is defined by canonical commutation relations \eqref{eq:commRel}, we have
	\begin{equation*}
		D = C- \frac12 J.
	\end{equation*}

	\begin{lemma}
		For a Gaussian state with zero mean and the matrix of second central moments $ D $, we have
		\begin{equation}\label{eq:recRel}
			\langle \mathfrak{a}_{\{1,2\} \sqcup I} \rangle = \sum_{j,k \in I, j< k}(D_{1j}  D_{2k} + D_{1k} D_{2j} )	\langle \mathfrak{a}_{I \setminus \{1,2, j, k\} } \rangle  + D_{12} \langle \mathfrak{a}_{I} \rangle.
		\end{equation}
	\end{lemma}
	
	\begin{proof}
		Let us calculate
		\begin{equation*}
			g_1^T\left(\frac{\partial}{\partial(i \mathbf{z})}- \frac12 J (i \mathbf{z})\right) e^{\frac12(i \mathbf{z})^T C (i \mathbf{z})} = g_1^T\left(C- \frac12 J \right) (i \mathbf{z})e^{\frac12(i \mathbf{z})^T C (i \mathbf{z})}  = g_1^T D (i \mathbf{z}) e^{\frac12(i \mathbf{z})^T C (i \mathbf{z})} ,
		\end{equation*}
		then
		\begin{align}
			g_2^T\left(\frac{\partial}{\partial(i \mathbf{z})}- \frac12 J (i \mathbf{z})\right)  g_1^T\left(\frac{\partial}{\partial(i \mathbf{z})}- \frac12 J (i \mathbf{z})\right) e^{\frac12(i \mathbf{z})^T C (i \mathbf{z})} \nonumber \\
			= g_2^T\left(\frac{\partial}{\partial(i \mathbf{z})}- \frac12 J (i \mathbf{z})\right)\left( g_1^T D (i \mathbf{z}) e^{\frac12(i \mathbf{z})^T C (i \mathbf{z})} \right) \nonumber \\
			= (g_1^T D g_2+  g_1^T D (i \mathbf{z}) g_2^T D (i \mathbf{z}))e^{\frac12(i \mathbf{z})^T C (i \mathbf{z})}. \label{eq:secOrdAppl}
		\end{align}
		And let us calculate the commutation relations
		\begin{align*}
			\biggl[ g_3^T\left(\frac{\partial}{\partial(i \mathbf{z})}- \frac12 J (i \mathbf{z})\right) , g_1^T D g_2+  g_1^T D (i \mathbf{z}) g_2^T D (i \mathbf{z}) \biggr] = \biggl[ g_3^T \frac{\partial}{\partial(i \mathbf{z})} ,  g_1^T D (i \mathbf{z}) g_2^T D (i \mathbf{z}) \biggr]\\
			=   g_1^T D g_3 g_2^T D (i \mathbf{z}) + g_1^T D (i \mathbf{z}) g_2^T Dg_3
		\end{align*}
		and
		\begin{align*}
			\biggl[ g_4^T\left(\frac{\partial}{\partial(i \mathbf{z})}- \frac12 J (i \mathbf{z})\right) ,  g_1^T D g_3 g_2^T D (i \mathbf{z}) + g_1^T D (i \mathbf{z}) g_2^T D g_3 \biggr]\\
			= g_1^T D g_3 g_2^T D g_4 + g_1^T D g_4 g_2^T D g_3 .
		\end{align*}
		Then by applying Eq.~\eqref{eq:secOrdAppl} and this  commutation relations we have
		\begin{align*}
			g_m^T\left(\frac{\partial}{\partial(i \mathbf{z})}- \frac12 J (i \mathbf{z})\right) \ldots g_1^T\left(\frac{\partial}{\partial(i \mathbf{z})}- \frac12 J (i \mathbf{z})\right)  e^{\frac12(i \mathbf{z})^T C (i \mathbf{z})}\\
			= g_m^T\left(\frac{\partial}{\partial(i \mathbf{z})}- \frac12 J (i \mathbf{z})\right) \ldots g_3^T\left(\frac{\partial}{\partial(i \mathbf{z})}- \frac12 J (i \mathbf{z})\right) (g_1^T D g_2+  g_1^T D (i \mathbf{z}) g_2^T D (i \mathbf{z}))e^{\frac12(i \mathbf{z})^T C (i \mathbf{z})} \\
			= (g_1^T D g_2+  g_1^T D (i \mathbf{z}) g_2^T D (i \mathbf{z}))  g_m^T\left(\frac{\partial}{\partial(i \mathbf{z})}- \frac12 J (i \mathbf{z})\right) \ldots g_3^T\left(\frac{\partial}{\partial(i \mathbf{z})}- \frac12 J (i \mathbf{z})\right) e^{\frac12(i \mathbf{z})^T C (i \mathbf{z})}\\
			+ \sum_{j,k} (g_1^T D g_k g_2^T D g_j + g_1^T D g_j g_2^T D g_k ) \prod_{l \neq 1,2, j,k} g_l^T\left(\frac{\partial}{\partial(i \mathbf{z})}- \frac12 J (i \mathbf{z})\right) e^{\frac12(i \mathbf{z})^T C (i \mathbf{z})} 
		\end{align*}
		
		Assuming $ \mathbf{z}=0 $ at both sides of these equations and taking into account corollary \ref{cor:meanMomets} we obtain Eq.~\eqref{eq:recRel}.
	\end{proof}
	
	In particular, for $ I = \{1,2,3,4\} $ we have
	\begin{equation*}
		\langle \mathfrak{a}_{1234} \rangle = D_{13}  D_{24} + D_{14} D_{23}   + D_{12} \langle \mathfrak{a}_{34} \rangle,
	\end{equation*}
	and taking into account $ \langle \mathfrak{a}_{34} \rangle =  D_{34}  $ we obtain
	\begin{equation*}
		\langle \mathfrak{a}_{1234} \rangle = D_{12}  D_{34} +D_{13}  D_{24} + D_{14} D_{23}.
	\end{equation*}
	And for $ I = \{1,2,3\} $ we have
	\begin{equation*}
		\langle \mathfrak{a}_{123} \rangle = D_{12}   \langle \mathfrak{a}_{3} \rangle =0.
	\end{equation*}
	
	In general, applying recurrence relation \eqref{eq:recRel} iteratively we obtain Isserlis - Wick theorem.
	
	\begin{theorem}\label{th:IsserlisWick}
		(Isserlis - Wick) 
		For the Gaussian state with zero mean and the matrix of second central moments $ D $, we have
		\begin{equation*}
			\langle a_I  \rangle = D_I,
		\end{equation*}
		where $ D_I $ is defined similarly to Eq.~\eqref{eq:psiGtBetaIDef} for $ \beta_{I}(t) $
		\begin{equation*}
			D_I = \begin{cases}
				\sum_{I = p_1 \sqcup \ldots \sqcup p_{|I|/2}} D_{p_1} \ldots D_{p_{|I|/2}}, & \text{for even }|I|,\\
				0, & \text{for odd }|I|,
			\end{cases}
		\end{equation*}
		and the sum is taken over all pairings of elements of $ I $.
	\end{theorem}

	\begin{corollary}
		For the Gaussian state with mean  $ \mu $ and the matrix of second central moments $ D $, we have
		\begin{equation}\label{eq:IsserlisWickNonZeroMean}
			\langle a_I  \rangle = \sum_{I = I_1 \sqcup I_2} \mu_{I_1} D_{I_2} .
		\end{equation}
	\end{corollary}
	
	\begin{proof}
		For  the Gaussian state with non-zero mean $ \mu $ Theorem \ref{th:IsserlisWick} can be applied to
		\begin{equation*}
			\prod_{k \in I}(a_k - \mu_k) =D_{I},
		\end{equation*}
		then
		\begin{equation*}
			\langle a_I  \rangle = \prod_{k \in I}(a_k - \mu_k + \mu_k)  = \sum_{I = I_1 \sqcup I_2} \prod_{k \in I_1}(a_k - \mu_k)  \prod_{l \in I_2}\mu_l = \sum_{I = I_1 \sqcup I_2} \mu_{I_1} D_{I_2} .
		\end{equation*}
	\end{proof}
	
	In particular, we have
	\begin{align*}
		\langle \mathfrak{a}_1  \rangle =& \mu_1, \\
		\langle \mathfrak{a}_{12}  \rangle =&  \mu_1 \mu_2 + D_{12},  \\
		\langle \mathfrak{a}_{123} \rangle =& D_{12} \mu_3 + D_{13} \mu_2 + D_{23} \mu_1  + \mu_1 \mu_2 \mu_3,\\
		\langle \mathfrak{a}_{1234} \rangle =&  \mu_1 \mu_2   + \mu_1 \mu_2   D_{34} + \mu_1 \mu_3   D_{24} + \mu_1 \mu_4  D_{23}, \\
		&+ \mu_3 \mu_4 D_{12}  + \mu_2 \mu_4 D_{13} + \mu_2 \mu_3 D_{14}  + D_{12}  D_{34} +D_{13}  D_{24} + D_{14} D_{23} .
	\end{align*}
	
	Now let us assume that the initial state is Gaussian, i.e. the moments of creation and annihilation operators satisfy Eq.~\eqref{eq:IsserlisWickNonZeroMean} and the evolution of $ \mathfrak{a}_I $ is defined by Eq.~\eqref{eq:HeisenbergSolution}. Then
	Eq. \eqref{eq:meanEvol} leads to
	\begin{equation*}
		\mu_I(t) = \prod_{k \in I} ( G_k(t) ( \mu_k(0) + \psi_k(t))) = G_{I}(t)  \sum_{I =I_1 \sqcup I_2} \mu_{I_1}(0) \psi_{I_2}(t)
	\end{equation*}
	and Eq. \eqref{eq:centMomEvol} leads to
	\begin{align*}
		D_{I}(t) =  \sum_{I = p_1 \sqcup \ldots \sqcup p_{|I|/2}} G_{p_1}(t) ( D_{p_1}(0) +  \beta_{p_1}(t)) \ldots G_{p_{|I|/2}}(t) ( D_{p_{|I|/2}}(0) +  \beta_{p_{|I|/2}}(t)) \\
		= G_{I}(t) \sum_{I =I_1 \sqcup I_2} D_{I_1}(0) \beta_{I_2}(t).
	\end{align*}
	Then by Eq.~\eqref{eq:HeisenbergSolution} we have
	\begin{align*}
		\langle \mathfrak{a}_{I} (t) \rangle &= G_{I}(t) \sum_{I = I_1 \sqcup I_2  \sqcup I_3 } \psi_{I_1} (t) \beta_{I_2} (t) \langle \mathfrak{a}_{I_3} (0) \rangle = G_{I}(t) \sum_{I = I_1 \sqcup I_2  \sqcup I_3 \sqcup I_4 } \psi_{I_1} (t) \beta_{I_2} (t)  \mu_{I_3}(0) D_{I_4}(0) \\
		&=  \sum_{I = I_1' \sqcup I_2' }  \left(G_{I_1'}(t) \sum_{ I_1' = I_1 \sqcup I_3 }\psi_{I_1} (t) \mu_{I_3}(0) \right)  \left(G_{I_2'}(t)\sum_{ I_2' = I_2 \sqcup I_4 } \beta_{I_2} (t)  D_{I_4}(0)\right) \\
		&= \sum_{I = I_1' \sqcup I_2'} \mu_{I_1'}(t) D_{I_2'}(t) = \sum_{I = I_1 \sqcup I_2} \mu_{I_1}(t) D_{I_2}(t)
	\end{align*}
	Hence, the dynamics preserves the Isserlis - Wick theorem. This is not surprising, since the dynamics with a quadratic generator preserves the Gaussian states \cite{Vanheuverzwijn1978, Heinosaari2010}.

	\section{Averaging with respect to Poisson process}
	\label{sec:Poisson}
	
	If one now averages time-evolution with respect to the Poisson process with the parameter $ \lambda > 0 $, then we obtain a new semigroup 
	\begin{equation*}
		\sum_{n=0}^{\infty} e^{ n\mathcal{L}_{H, \Gamma, f}^* } \frac{(\lambda t)^n}{n!} e^{- \lambda t}= e^{\lambda t (  e^{\mathcal{L}_{H, \Gamma, f}^* } - 1)}
	\end{equation*}
	with the generator $  e^{\mathcal{L}_{H, \Gamma, f}^* } - 1 $. Thus, we have the following Heisenberg equation
	\begin{equation}\label{eq:HeisPoisson}
		\frac{d}{dt} X(t) = \lambda (e^{\mathcal{L}_{H, \Gamma, f}^*} X(t)  - X(t) ).
	\end{equation}
	
	For  $ e^{\mathcal{L}_{H, \Gamma, f}^* }  \mathfrak{a}_{I} (t)  $ can be calculated by Theorem \ref{th:HeisSol}, which leads to the following theorem.
	
	\begin{theorem}
		For $ X(t) = \mathfrak{a}_{I} (t) $ Eq. \eqref{eq:HeisPoisson} takes the form
		\begin{equation}\label{eq:PoissonHeisenberg}
			\frac{d}{dt}\mathfrak{a}_{I} (t) = \lambda (( G_{I}(1) -1) \mathfrak{a}_{I} (t)+ G_{I}(1) \sum_{I = I_1 \sqcup I_2  \sqcup I_3, I_3 \neq I} \psi_{I_1} (1) \beta_{I_2} (1) \mathfrak{a}_{I_3} (t)),
		\end{equation}
		where $ G_{I}(t) $, $  \psi_{I_1} (1)$, $ \beta_{I_2} (1) $ are defined by Eqs.~\eqref{eq:GtConst}--\eqref{eq:XitConst}.
	\end{theorem}
	
	In particular, for $ I = \{1\} $ and $ I = \{1, 2\} $ we have
	\begin{align*}
		\frac{d}{dt} \mathfrak{a}_{1} (t) &=  (e^{B_1} - 1)\mathfrak{a}_{1} (t) + \frac{e^{B_1 } - 1}{B_1} \varphi_1,\\
		\frac{d}{dt}\mathfrak{a}_{12} (t) &=  (e^{B_1 + B_2}-1)  \mathfrak{a}_{12} (t) + e^{B_1}\mathfrak{a}_{1} (t) \frac{e^{B_2} - 1}{B_2} \varphi_{2} + \frac{e^{B_1} - 1}{B_1} \varphi_{1}  e^{B_2}\mathfrak{a}_{2} (t) + \frac{e^{B_1 + B_2} - 1}{B_1 + B_2} \beta_{12}.
	\end{align*}
	By averaging these equations with respect to the initial density matrix one obtains
	\begin{align*}
		\frac{d}{dt} D_{12}(t) =& (e^{B_1 + B_2}-1) D_{12}(t)   + \frac{e^{B_1 + B_2} - 1}{B_1 + B_2} \beta_{12} \\
		&+ (e^{B_1} - 1) (e^{B_2} - 1) \mu_{1} (t) \mu_{2} (t) +
		(e^{B_1} -1)\mu_{1} (t) \frac{e^{B_2} - 1}{B_2} \varphi_{2} + \frac{e^{B_1} - 1}{B_1} \varphi_{1}  (e^{B_2} -1)\mu_{2} (t).
	\end{align*}
	Thus, similarly to \cite{Linowski2021} the equation for $ D_{12}(t) $ is not closed and depends on the first moments $ \mu(t) $. Hence, such evolution does not preserve the Isserlis-Wick theorem. Actually, it is not a surprise due to the fact that only the GKSL generators leading to preservation of Gaussian states during evolution are quadratic ones \cite[Prop.~4]{Heinosaari2010}. 
	
	Eq.~\eqref{eq:PoissonHeisenberg} can be solved as
	\begin{equation*}
		\mathfrak{a}_{I} (t) = e^{\lambda ( G_{I}(1) -1) t} \mathfrak{a}_{I} (0)+ \lambda \int_0^t e^{\lambda ( G_{I}(1) -1)(t-\tau)} G_{I}(1) \sum_{I = I_1 \sqcup I_2  \sqcup I_3, I_3 \neq I} \psi_{I_1} (1) \beta_{I_2} (1) \mathfrak{a}_{I_3} (\tau) d\tau,
	\end{equation*}
	where $  \mathfrak{a}_{I_3} (\tau) $ at the right-hand side of this equation has $ I_3 $ of lower cardinality than $ I $ This allows one to solve the equation for $ \mathfrak{a}_{I}(t) $ with $ I $ of lower cardinality and substitute them into equations for $ \mathfrak{a}_{I}(t) $ with $ I $ of higher cardinality and obtain solutions for any given $ I $.
	
	Similarly to \cite{NosTer20, Linowski2021} one can consider a master equation of the form
	\begin{equation*}
		\frac{d}{dt} X(t) = \sum_k \lambda_k \left(e^{\mathcal{L}_{H^{(k)}, \Gamma^{(k)}, f^{(k)}}^*} X(t)  - X(t) \right), \qquad \lambda_k > 0,
	\end{equation*}
	i.e. with the generator, which is a combination of the generators of \eqref{eq:HeisPoisson}. Then, analogously to Eq.~\eqref{eq:PoissonHeisenberg}, one has for $ X(t) = \mathfrak{a}_{I} (t) $ the following Heisenberg equation
	\begin{equation}\label{eq:multiPoissonHeisenberg}
		\frac{d}{dt}\mathfrak{a}_{I} (t) =  \sum_k \lambda_k (( G_{I}^{(k)}(1) -1) \mathfrak{a}_{I} (t)+ G_{I}^{(k)}(1) \sum_{I = I_1 \sqcup I_2  \sqcup I_3, I_3 \neq I} \psi^{(k)}_{I_1} (1) \beta^{(k)}_{I_2} (1) \mathfrak{a}_{I_3} (t)),
	\end{equation}
	where  $ G_{I}^{(k)}(t) $, $  \psi_{I_1}^{(k)} (1)$, $ \beta_{I_2}^{(k)} (1) $ are defined similarly to Eqs.~\eqref{eq:GtConst}--\eqref{eq:XitConst}.
	
	After averaging with respect to the initial density matrix both Eq.~\eqref{eq:PoissonHeisenberg} and Eq.~\eqref{eq:multiPoissonHeisenberg} allow one to obtain any finite-order moment dynamics by solving a system of a finite number of  linear ordinary differential equations. Let us also remark that analogously to \cite{NosTer20} Heisenberg equations are enough to define multi-time correlations functions by the regression formula.
	
	\section{Conclusions}
	
	Similarly to the case of averaging unitary dynamics with respect to the Levy  processes and fields  \cite{Davies1972,Kossakowski1972, Kummerer87, Holevo1996, Holevo98} the GKSL equations arising from averaging with respect to the Possion process are a key ingredient for generalization of our result to the arbitrary  Levy  processes and fields due to the Levy-Khintchine theorem \cite{Sato2001}. The other important ingredient is the GKSL equations arising  from averaging with respect to the Wiener process which seems to be also manageable by methods developed here similarly to \cite{Ivanov2022}.
	
	The most of the results of this work can be generalized for the fermionic case. But what is more important is that they could be generalized for the case, when lemma \ref{lem:LeibnizType} leads to the closed Heisenberg equations for some operators only, but not all the moments of creation and annihilation operators. It could be interpreted as some kind of dissipative analog of dynamical symmetry \cite{Malkin1979}. And averaging with respect to the Poisson processes or even the arbitrary Levy  processes seems to preserve such a symmetry. We think that it is an important direction of further development.

\end{document}